\documentclass[aps,  amsmath,  amssymb,  amsfonts,  showpace,  twocolumn,  pra]{revtex4}

\usepackage{diagbox}
\usepackage{amscd}
\usepackage{amsthm}
\usepackage{graphicx}
\usepackage{epstopdf}
\usepackage{mathdots}
\usepackage{blkarray}
\makeatletter
\newbox\BA@first@box
\makeatother

\newtheorem{Thm}{Theorem}
\newtheorem{Lem}{Lemma}
\newtheorem{Cor}{Corollary}
\newtheorem{Prop}{Proposition}
\newtheorem{Exm}{Example}

\usepackage[
  colorlinks,
  linkcolor = blue,
  citecolor = blue,
  urlcolor = blue]{hyperref}

\usepackage{amsmath,  amssymb,  amsthm}

\usepackage{tikz}
\usetikzlibrary{arrows}
\usetikzlibrary{calc,  positioning}

\usepackage{thm-restate}

\usepackage{colonequals}

\usepackage[font=small]{caption}

\newtheoremstyle{noCaption}
{\topsep}
{\topsep}
{\itshape}
{}
{}
{}
{0pt}
{}%

\begin{document}
\title{{\color{black} Small sets of genuinely nonlocal GHZ states in multipartite systems} }
\author{Zong-Xing Xiong$^{1}$\footnote{hongbe123@outlook.com}, Yongli Zhang$^{2}$, Mao-Sheng Li$^{3}$, Lvzhou Li$^{1}$\footnote{lilvzh@mail.sysu.edu.cn}\\
{\footnotesize{\it $^1$ Institute of Quantum Computing and Software, School of Computer Science and Engineering,}}\\
{\footnotesize{\it Sun Yat-Sen University, Guangzhou 510006, China}}\\
{\footnotesize{\it $^2$ School of Mathematics and Systems Science, Guangdong Polytechnic Normal University, Guangzhou 510665, China}}\\
{\footnotesize{\it $^3$ School of Mathematics, South China University of Technology, Guangzhou 510641, China}}\\
}

\affiliation{
}

\begin{abstract}
{\color{black} A set of orthogonal multipartite quantum states are called (distinguishability-based) genuinely nonlocal if they are locally indistinguishable across any bipartition of the subsystems. In this work, we consider the problem of constructing small genuinely nonlocal sets consisting of generalized GHZ states in multipartite systems. For system $(\mathbb{C}^2)^{\otimes N}$ where $N$ is large, using the language of group theory, we show that a tiny proportion $\Theta(1/\sqrt{2^N})$ of states among the $N$-qubit GHZ basis suffice to exhibit genuine nonlocality. Similar arguments also hold for the canonical generalized GHZ bases in systems $(\mathbb{C}^d)^{\otimes N}$, wherever $d$ is even and $N$ is large. What is more, moving to the condition that any fixed $N$ is given, we show that $d+1$ genuinely nonlocal generalized GHZ states exist in $(\mathbb{C}^d)^{\otimes N}$, provided the local dimension $d$ is sufficiently large. As an additional merit, within and beyond an asymptotic sense, the latter result also indicates some evident limitations of the ``trivial othogonality--preserving local measurements'' (TOPLM) technique that has been utilized frequently for detecting genuine nonlocality. }

\end{abstract}
\pacs{03.  67.  Hk,  03.  65.  Ud }
\maketitle

\section{Introduction} \label{section1}

Quantum nonlocality, as one of the most fascinating phenomenon in quantum mechanics, is most well-known to the public in the manifestation of Bell nonlocality \cite{Bell1964,Brunner2014}. To exhibit such form of nonlocality, the existence of quantum entanglement is always necessary. However, there are also other forms of nonlocality, including one that is based on local distinguishability of multipartite quantum states.  It serves to explore the fundamental properties about locally accessing global information and unlike Bell nonlocality, entanglement is no more necessarily required. Such nonlocality was first revealed by Bennett and his co-workers \cite{Bennett99}, who constructed nine orthogonal product states in $\mathbb{C}^3 \otimes \mathbb{C}^3$ that are not perfectly distinguishable by the two separated parties, provided only local operations and classical communication (LOCC) are allowed. Since this seminal work, such form of distinguishability-based nonlocality has been studied extensively (see \cite{Bennett1999,Walgate00,Walgate02,Ghosh01,Ghosh02,DiVicenzo03,Horodecki03,Ghosh04,Fan04,Hayashi06,Nathanson05,Bandyopadhyay11,Yu12,Nathanson13,Cosentino13,CosentinoR14,Tian15,Tian2015,Li15,Yu15,Singal17,Xu16, Zhang17, Halder18, Jiang20, Zhen22, Cao23} for an incomplete list).

Probabilisticly, a set of more quantum states might usually tend to be harder for distinguishing while a set with less states is often more likely to be distinguishable. Informally speaking, the reason is that all supersets of an indistinguishable set must be indistinguishable,  whereas all subsets of a distinguishable one are also distinguishable. Therefore, to get some nontrivial knowledge about such form of nonlocality, one is either interested in the maximal number of states that retain locality, or interested in the minimal number of states that exhibit nonlocality. For two-partite systems, Hayashi et al. discovered that the maximal number of orthogonal pure states that are locally distinguishable cannot exceed the total dimension over the average entanglement of the states \cite{Hayashi06}. This result, to some extent, indicates that quantum states with more entanglement are generally (while not always) more difficult to be distinguished. \ Inspite of such a quantitative upper bound in the direction of retaining locality, we have, on the other hand, little idea about the minimal number of states that can exhibit nonlocality. In the literature, lots of efforts have been paid in seeking small locally indistinguishable sets consisting of maximally entangled states \cite{Nathanson05,Bandyopadhyay11,Yu12,Nathanson13,Cosentino13,CosentinoR14,Tian15,Tian2015,Li15,Yu15}, which are widely believed to have stronger nonlocality than any other bipartite quantum states. Frustratingly, even a simple problem whether there exist three locally indistinguishable maximally entangled states in $\mathbb{C}^d \otimes \mathbb{C}^d$ for $d > 3$ remains unsolved today. What is known is that any two orthogonal pure states can always be locally distinguished, no matter whether the states are entangled or not \cite{Walgate00}.  Apart from these, in the multipartite scenarios, there are also a series of works exploring small locally indistinguishable sets of multipartite product states \cite{Xu16, Zhang17, Halder18, Jiang20, Zhen22, Cao23}. In the most recent work \cite{Cao23} particularly, Cao et al. showed the existence of $d + 1$ locally indistinguishable product states in $(\mathbb{C}^d)^{\otimes N}$, outperforming results of the others. Nevertheless, in the context of multipartite (distinguishability-based) nonlocality, there are other stronger paradigms that have been proposed.

A set of orthogonal multipartite quantum states are called genuinely nonlocal if they are locally indistinguishable across any bipartition of the subsystems. Surely, genuine nonlocality derives local indistinguishability immediately and it is by definition a much stronger form of nonlocality than local indistinguishability. In \cite{Rout19, Halder19, Li21,Rout21, shi, xiong23}, several techniques have been applied for constructing genuinely nonlocal sets with different types of multipartite quantum states. {\color{black} Among them, the so called ``trivial orthogonality--preserving local measurements'' (TOPLM, whose definition shown in Section \ref{section2}), which was originated from Walgate and Hardy \cite{Walgate02}, is the most widely applicable one. In \cite{Halder19}, such technique was first utilized by Halder et al. for deriving genuinely nonlocal sets of product states in several small systems. As a matter of fact, the notion put forward by Halder et al. is ``strong nonlocality'', which refers to local irreducibility of a set of multipartite quantum states through any bipartition of the subsystems. A set of orthogonal multipartite quantum states are called locally irreducible if it is impossible to eleminate one or more states from the whole set, with restriction that only orthogonality--preserving local measurements (OPLM) are allowed. In practice however, it is often difficult to determine whether a set of states are locally irreducible, except those cases where only ``trivial orthogonality--preserving local measurements'' can be performed by the subsystems. By far, though lots of efforts have been paid in seeking strongly nonlocal sets in the literature (see \cite{Zhang19, Yuan20, Shi20, Wang21, Shi2022, Shi22, Zhou22, Zhou23, Li23} for an incomplete list), all existing examples are constructed through the TOPLM technique \cite{strongest nonlocality}.  Recently, it has been proved by Li and Wang that in system $(\mathbb{C}^d)^{\otimes N}$, all strongly nonlocal sets constructed in this way must have cardinality no smaller than $d^{N - 1} + 1$ \cite{Li23}. Notice that by definition, local irreducibility is nothing other than a sufficient condition for deriving local indistinguishability and arguably, so is strong nonlocality with respect to genuine nonlocality.}

{\color{black}
In this work, we study the problem of constructing small genuinely nonlocal sets consisting of generalized GHZ states in multipartite system. We first consider systems $(\mathbb{C}^d)^{\otimes N}$ on condition that the $d$ is fixed and  $N$ is large. For the $N$-qubit GHZ bases, we show that a tiny proportion {\small $\Theta(1/\sqrt{2^{N}})$} of the states among such bases suffice to exhibit genuine nonlocality. Similar arguments also hold for the canonical generalized GHZ bases in $(\mathbb{C}^d)^{\otimes N}$, wherever the local dimension $d$ are even. As for the case where $N$ is fixed and $d$ is large, we show the existence of $d+1$ genuinely nonlocal generalized GHZ states in $(\mathbb{C}^d)^{\otimes N}$. We argue that, within and beyond an asymptotic sense, this result also indicates some evident limitations of the TOPLM techique for detecting small genuinely nonlocal sets.}

The rest of this paper is organized as follows: Section \ref{section2} provides some relevant definitions and notations. In Section \ref{section3}, we investigate genuine nonlocality of the canonical generalized GHZ bases in $(\mathbb{C}^d)^{\otimes N}$, in case that $d$ is even. In Section \ref{section4}, we demonstrate how to construct $d+1$ genuinely nonlocal generalized GHZ states in $N$-partite systems $(\mathbb{C}^d)^{\otimes N}$. We draw our conclusion and discuss some relating problems in Section \ref{section5}.

\section{preliminaries} \label{section2}

In this section, we provide some relevant definitions and notations that is necessary in this paper.

\smallskip
\textit{Local distinguishability (LOCC-distinguishability)}:
\ A set of orthogonal multi-partite quantum states, which is priorly known to several spatially separated parties, is said to be locally distinguishable if the parties are able to tell exactly which state they share through some protocols, provided only local operations (measurements) and classical communications (LOCC) are allowed.

\medskip
\textit{Local irreducibility}:
A set of multipartite orthogonal quantum states is said to be locally irreducible if it is not possible to locally eliminate one or more states from the set while preserving orthogonality of the postmeasurement states. Typical examples of locally irreducible sets include the two-qubit Bell basis and the $N$-qubit GHZ basis \cite{Halder19}.

\medskip
\textit{Genuine nonlocality}:
\ A set of orthogonal multipartite quantum states is called (distinguishability-based) {\em genuinely nonlocal} if the states are locally indistinguishable across any bipartition of the subsystems.

\medskip
\textit{Strong nonlocality}:
\ A set of orthogonal multipartite quantum states is called {\em strongly nonlocal} if the states are locally irreducible across any bipartition of the subsystems.

\medskip
{\color{black}
\textit{Trivial orthogonality--preserving local measurements (TOPLM)}: \ A measurement is called nontrivial if not all the POVM elements are proportional to the identity. Otherwise, we call the measurement trivial. \ In any local discrimination protocol, one of the parties must go first and whoever goes first must be able to perform some nontrivial orthogonality--preserving local measurements. Therefore, for a set of orthogonal multipartite quantum states, \ if only trivial orthogonality--preserving local measurements (TOPLM) can be performed by each of the parties, then the states must be locally indistinguishable (irreducible). In recent literature, TOPLM has been frequently utilized as a technique for detecting genuinely nonlocal (strongly nonlocal) sets \cite{Halder19,Zhang19, Yuan20, Shi20, Wang21, Shi2022, Shi22, Zhou22, Zhou23, Li23}.
}


\medskip
\textit{PPT-distinguishability}:
\ In the literature, since the mathematical structure of LOCC measurements are rather complicated, they are usually approximated by separable measurements \cite{Bandyopadhyay15} or PPT (positive-partial-transpose) measurements \cite{Yu14}. A positive semidefinite operator $0 \leq M \in \mathrm{Pos}(\mathcal{H}_\mathrm{A} \otimes \mathcal{H}_\mathrm{B})$ is called ``PPT'', if its partial transpose about one subsystem (say, $\mathrm{A}$) is also positive semidefinite: $M^{\mathrm{T_A}} \geq 0$. Since LOCC measurement opeartors (LOCC-POVMs) are separable and separable operators are PPT, \ PPT-indistinguishability implies local indistinguishability immediately.

\medskip
\textit{$N$-partite generalized GHZ states}:
\ In $N$-partite system $\mathcal{H}_{\mathrm{A}_1} \otimes \mathcal{H}_{\mathrm{A}_2} \otimes \cdots \otimes \mathcal{H}_{\mathrm{A}_N}$ where $\mathcal{H}_{\mathrm{A}_1} = \mathcal{H}_{\mathrm{A}_2} = \cdots = \mathcal{H}_{\mathrm{A}_N} = \mathbb{C}^d$, \ quantum states like
\begin{equation*}
\frac{1}{\sqrt{d}} \sum_{j=0}^{d-1} \ \left|\zeta_j^{(1)} \zeta_j^{(2)} \cdots \zeta_j^{(N)}\right>_{\mathrm{A}_1 \mathrm{A}_2 \cdots \mathrm{A}_N}
\end{equation*}
are called {\em generalized GHZ states},  where each $\{|\zeta_j^{(n)}\rangle\}_{j=0}^{d-1}$ is an arbitrary set of orthogonal basis for the subsystem $\mathcal{H}_{\mathrm{A}_n}$ $(1 \leq n \leq N)$ .

\medskip
In fact, the discussion about genuine nonlocality of the GHZ states can be dated back to Hayashi et al.'s work \cite{Hayashi06}. Here we reproduce their result with the following lemma:

\begin{Lem}\label{lemma1}\textnormal{\cite{Hayashi06}}
\ In $N$-partite system $(\mathbb{C}^d)^{\otimes N}$, any $s \geq d^{N-1} + 1$  orthogonal generalized GHZ states are genuinely nonlocal.
\end{Lem}

In other words, $d^{N-1}$ is the maximal number of orthogonal generalized GHZ states that retain locality --- locally distinguishable in at least one bipartition. In particular,  for the multi-qubit case (namely, when $d=2$), Bandyopadhyay showed further that this upper bound is tight \cite{Bandyopadhyay10}. On the other side however, what is the minimal number of states that can exhibit genuine nonlocality?  In the simplest multi-qubit case, for either $N = 2$ or $N = 3$, such a minimal number turns out to be $2^{N-1} + 1$ ({\color{black} For the case $N=2$, where genuine nonlocality is just local indistinguishability by definition and the ``two-qubit GHZ basis'' is nothing other than the Bell basis in $\mathbb{C}^{2} \otimes \mathbb{C}^{2}$, such a statement is obvious; For $N=3$, it follows from Proposition 1 of \cite{xiong23}}). That is, the upper bound and lower bound ``encounter'' here. Is this still true for cases $N > 3$? Besides, what is the situation when $d > 2$?  Such problems will be addressed in the following sections.

\bigskip

\section{ Genuine nonlocality for the canonical generalized GHZ bases}\label{section3}

In this section, we investigate genuine nonlocality of a special form of generalized GHZ states: the canonical generalized GHZ bases, whose definition is shown below.

\smallskip
\textit{Canonical generalized GHZ basis}:
In $N$-partite system $(\mathbb{C}^{d})^{\otimes N}$ where $\{|0\rangle, \cdots, |d-1\rangle\}$ is the standard orthogonal basis for each subsystem, the $d^N$ orthogonal generalized GHZ states
\begin{equation}\label{canonical}
\frac{1}{\sqrt{d}} \sum_{j=0}^{d-1} \omega_{d}^{j(k-1)}  \left|j, j\oplus i_{2}, \phantom{^{^{1}}} \cdots , j\oplus i_{N}\right>_{\mathrm{A}_1\mathrm{A}_2\cdots\mathrm{A}_N} \ \ \ \ \ \
\end{equation}
where $k \in \{1, \cdots, d\}$ and $i_2, \cdots, i_{N} \in \{0, \cdots, d-1\}$ consist a set of basis for the global system. Here, $\omega_d = e^{\frac{2\pi i}{d}}$ and ``$\oplus$'' is the ``mod $d$'' addition. They are called the {\em canonical generalized GHZ basis} and states in such form are called {\em generalized GHZ states in canonical form}. Notice that for $d =2$,  such definition coincides with the ordinary ``$N$-qubit GHZ basis''.

\medskip
In what follows, we show that for such bases, just a tiny proportion {\small $\Theta[1/(\frac{d}{\sqrt{2}})^{N}]$} of the whole set suffices to exhibit genuine nonlocality, on condition that $N$ grows large. Herein, we write $f(n) = \Theta[g(n)]$ if there exist positive constants $n_0$ and $c_1, c_2$ such that for $n \geq n_0$, we have $c_1 \cdot g(n) \leq f(n) \leq c_2 \cdot g(n)$. For convenience of explanation, we divide our discussion into three parts.
\\

\begin{table*}[t]
\caption{ \ Genuine nonlocality for subset $\mathcal{N}_7^{(4)}$ of the four-qubit GHZ basis}\label{tb1}
\begin{tabular}{|c|c|c|c|c|}
\hline
                  & $\psi_{0000, 1111}$ & $\psi_{0010}$ & $\psi_{0100}$ & $\psi_{1000}$ \\
\hline
$\psi_{0000, 1111}$ & {\footnotesize \bf null} & ABD$|$C       &  ACD$|$B      &  BCD$|$A      \\
\hline
$\psi_{0001, 1110}$ & ABC$|$D             & AB$|$CD       &  AC$|$BD      &  BC$|$AD      \\
\hline
\end{tabular}
\end{table*}
\begin{table*}[t]
\caption{ \ Genuine nonlocality for subset $\mathcal{N}_{11}^{(5)}$ of the five-qubit GHZ basis}\label{tb2}
\begin{tabular}{|c|c|c|c|c|c|c|c|c|}
\hline
   & $\phi_{00000, 11111}$ & $\phi_{00010}$ & $\phi_{00100}$ & $\phi_{01000}$ & $\phi_{10000}$ & $\phi_{00110}$ & $\phi_{01010}$ & $\phi_{10010}$ \\
\hline
$\phi_{00000, 11111}$ & {\footnotesize \bf null} &  ABCE$|$D  &  ABDE$|$C  &  ACDE$|$B  &  BCDE$|$A     &   ABE$|$CD     &   ACE$|$BD     &   BCE$|$AD     \\
\hline
$\phi_{00001, 11110}$ &   ABCD$|$E        &  ABC$|$DE  &  ABD$|$CE  &  ACD$|$BE  &  BCD$|$AE     &   AB$|$CDE     &   AC$|$BDE     &   BC$|$ADE     \\
\hline
\end{tabular}
\end{table*}

\medskip

\centerline{\bf 1. A straightforward construction}
\medskip

For the canonical generalized GHZ basis (\ref{canonical}) in system $(\mathbb{C}^{d})^{\otimes N}$ ($d \geq 2$), there is a quite straightforward construction of genuinely nonlocal subset that is shown below.
\begin{Prop}\label{1d}
\ In $N$-partite system $(\mathbb{C}^{d})^{\otimes N}$, the following subset of the canonical generalized GHZ basis (\ref{canonical}):
\begin{equation*}\begin{split}
& |\Gamma_k^{(\ast)}\rangle = \frac{1}{\sqrt{d}} \sum_{j=0}^{d-1} \ \omega_d^{j(k-1)} \ |j j \cdots j\rangle_{\mathrm{A}_1\mathrm{A}_2 \cdots \mathrm{A}_N},  \\
& \phantom{|\Gamma_k^{(\ast)}\rangle \sum_{j=0}^{d-1} \omega_d^{kj} {\mathrm{A}_1\mathrm{A}_2 \cdots \mathrm{A}_N} } \ \  (\omega_d = e^{2\pi i/d}; \  k = 1, \cdots, d)\\
& |\Gamma^{S}\rangle = \frac{1}{\sqrt{d}} \sum_{j=0}^{d-1} \ |j \cdots j\rangle_{S} |j\oplus 1 \cdots j \oplus 1\rangle_{\overline{S}}, \\
& \phantom{|\Gamma_k^{(\ast)}\rangle \sum_{j=0}^{d-1} \omega_d^{kj}|j\rangle_{\mathrm{A}_1} \otimes } (S \subset \{\mathrm{A}_1, \cdots, \mathrm{A}_N\}, \ 1 \leq |S| < \frac{N}{2}) \\
& \text{and (when $N$ is even)}\\
& |\Gamma^{R}\rangle = \frac{1}{\sqrt{d}} \sum_{j=0}^{d-1} \ |j \cdots j\rangle_{R} |j\oplus 1 \cdots j \oplus 1\rangle_{\overline{R}}, \\
& \phantom{|\Gamma_k^{(\ast)}\rangle \sum_{j=0}^{d-1} \omega_d^{kj}|j\rangle_{A_1} \ } (\mathrm{A}_1 \in R \subset \{\mathrm{A}_1, \cdots, \mathrm{A}_N\}, \ |R| = \frac{N}{2})\\
\end{split}\end{equation*}
is genuinly nonlocal. The cardinality of this subset is $d+2^{N-1}-1$.
\end{Prop}

\begin{proof}
It's not hard to verify the orthogonality of these states. In the bipartition $S|\overline{S}$ where $S$ is a nonempty subset of $\{\mathrm{A}_1, \cdots, \mathrm{A}_N\}$ such that $|S| < N/2$, the states $$|\Gamma^{S}\rangle \ \text{ and } \ |\Gamma_k^{(\ast)}\rangle \ (k = 1, \cdots, d)$$ constitute $d+1$ maximally entangled states in the $d \otimes d$ subspace $\mathcal{H}^{S}_d \otimes \mathcal{H}^{\overline{S}}_d \subset (\mathbb{C}^d)^{\otimes |S|} \otimes (\mathbb{C}^d)^{\otimes |\overline{S}|} $, where $\mathcal{H}^{S}_d =$ span$\{|j\cdots j\rangle_{S}\}_{j=0}^{d-1}$ and $\mathcal{H}^{\overline{S}}_d =$ span$\{|j\cdots j\rangle_{\overline{S}}\}_{j=0}^{d-1}$. Since any $d+1$ maximally entangled states in $\mathbb{C}^d \otimes \mathbb{C}^d$ are locally indistinguishable, the states $|\Gamma^{S}\rangle,$ $|\Gamma_k^{(\ast)}\rangle \ (1 \leq k \leq d)$ and hence the whole set of states are indistinguishable through $S|\overline{S}$. When $N$ is an even number, the same holds for the bipartitions $R|\overline{R}$, where $\mathrm{A}_1 \in R \subset \{\mathrm{A}_1, \cdots, \mathrm{A}_N\}$ and $|R| = N/2$. As a result, the whole set of $$d + C_{N}^{1} + \cdots + C_{N}^{\lfloor N/2 \rfloor} = d+2^{N-1}-1$$ states above are genuinely nonlocal.
\end{proof}

\smallskip

Nevertheless, such a result is by no means satisfactory because when $d = 2$, it tells nothing more than Lemma \ref{lemma1}. The question still remains: does there exist any genuinely nonlocal subset of the $N$-qubit GHZ basis with cardinality smaller than $2^{N-1} + 1$?  For $N=2$ and $N=3$, the answers turn out to be negative as mentioned in Section \ref{section2}. With no doubt, it's both reasonable and interesting to consider the same question for $N > 3$: do the negative answers still hold?

As it is shown in this section, the answer to the above question is no: there do exist genuinely nonlocal subsets of the $N$-qubit GHZ basis with size smaller than $2^{N-1} + 1$, whenever $N > 3$. In the next subsection, we will first provide some examples for the few-qubit cases.  The essence behind these examples will be explained in the third subsection, using the language of group theory. As it will be shown, when $N$ grows large, only $\Theta(\sqrt{2^N})$ states among the $N$-qubit GHZ basis (namely, a proportion $\Theta(1/\sqrt{2^N})$ of the whole set) suffice to exhibit genuine nonlocality.  Also, for $N$-qudit cases where $N$ is large and local dimension $d$ is even, a same argument holds: $\Theta(\sqrt{2^N})$ states (namely, a proportion {\small $\Theta[1/(\frac{d}{\sqrt{2}})^{N}]$}) of the canonical GHZ basis (\ref{canonical}) suffice to exhibit genuine nonlocality in  $(\mathbb{C}^d)^{\otimes N}$.


\begin{table*}[hbt]
\caption{ \ Genuine nonlocality for subset $\mathcal{N}_{15}^{(6)}$ of the six-qubit GHZ basis}\label{tb3}
\begin{tabular}{|c|c|c|c|c|c|c|c|c|}
\hline
   & $\varphi_{000000, 111111}$ & $\varphi_{000100}$ & $\varphi_{001000}$ & $\varphi_{010000}$ & $\varphi_{100000}$ & $\varphi_{001100}$ & $\varphi_{010100}$ & $\varphi_{100100}$ \\
\hline
$\varphi_{000000, 111111}$ &  {\scriptsize \bf ABCDEF$|\emptyset$} &  ABCEF$|$D  &  ABDEF$|$C  &  ACDEF$|$B  &  BCDEF$|$A  &  ABEF$|$CD &  ACEF$|$BD  &  BCEF$|$AD     \\
\hline
$\varphi_{000001, 111110}$ &   ABCDE$|$F        &  ABCE$|$DF  &  ABDE$|$CF  &  ACDE$|$BF  &  BCDE$|$AF  &  ABE$|$CDF &  ACE$|$BDF  &  BCE$|$ADF     \\
\hline
$\varphi_{000010, 111101}$ &   ABCDF$|$E        &  ABCF$|$DE  &  ABDF$|$CE  &  ACDF$|$BE  &  BCDF$|$AE  &  ABF$|$CDE &  ACF$|$BDE  &  BCF$|$ADE     \\
\hline
$\varphi_{000011, 111100}$ &   ABCD$|$EF        &  ABC$|$DEF  &  ABD$|$CEF  &  ACD$|$BEF  &  BCD$|$AEF  &  AB$|$CDEF &  AC$|$BDEF  &  BC$|$ADEF     \\
\hline
\end{tabular}
\end{table*}

\bigskip

\bigskip

\centerline{\bf 2. Examples for $N \leq 6$}
\bigskip

In this subsection, we are showing some examples of genuinely nonlocal subsets for the $N$-qubit GHZ basis that have cardinality smaller than $2^{N-1}+1$, for cases $N = 4, 5, 6$. Here and after, without additional specification, we always use ``$\mathrm{A}$'' to signify the first subsystem, ``$\mathrm{B}$'' to signify the second subsystem, ``$\mathrm{C}$'' to signify the third subsystem and so on. {\color{black} We also call a subset containing $k$ states a ``$k$-ary'' subset for abbreviation.}

\begin{Exm}\label{exm1}
\textnormal{
{(i)} In four-qubit system $(\mathbb{C}^2)^{\otimes 4}$, the four-qubit GHZ basis is defined as
\small
\begin{equation*}\begin{split}
&|\psi_{0,15}\rangle=\frac{|0000\rangle \pm |1111\rangle}{\sqrt{2}}, \ |\psi_{1,14}\rangle=\frac{|0001\rangle \pm |1110\rangle}{\sqrt{2}},\\
&|\psi_{2,13}\rangle=\frac{|0010\rangle \pm |1101\rangle}{\sqrt{2}}, \ |\psi_{3,12}\rangle=\frac{|0011\rangle \pm |1100\rangle}{\sqrt{2}},\\
&|\psi_{4,11}\rangle=\frac{|0100\rangle \pm |1011\rangle}{\sqrt{2}}, \ |\psi_{5,10}\rangle=\frac{|0101\rangle \pm |1010\rangle}{\sqrt{2}},\\
&|\psi_{6,9} \rangle=\frac{|0110\rangle \pm |1001\rangle}{\sqrt{2}}, \ \ \ |\psi_{7,8}\rangle=\frac{|0111\rangle \pm |1000\rangle}{\sqrt{2}},
\end{split}\end{equation*}\normalsize
where we assign the conjugate pairs with subscript indices summing up to $2^4-1$. For simplicity of comprehension, here and after, we also write the subscript indexes of the multi-qubit GHZ bases in their binary form. In this notation, we have
\begin{equation}\label{GHZ16}
|\psi_{0,15}\rangle = |\psi_{0000,1111}\rangle, \ \ \cdots \ , \ |\psi_{7,8}\rangle = |\psi_{0111,1000}\rangle.
\end{equation}
In what follows, for simplicity of representation, we will sometimes omit the ``$| \ \cdot \ \rangle$'' notation for some of the quantum states, which won't cause any ambiguity. Now we show that the  $7$-ary subset
$$\mathcal{N}_7^{(4)} = \{\psi_{0000,1111}, \ \psi_{0001,1110}, \ \psi_{0010}, \ \psi_{0100}, \ \psi_{1000}\}$$
of the GHZ basis (\ref{GHZ16}) is genuinely nonlocal:
\smallskip
\\ \indent(1) For the two conjugate pairs $\{\psi_{0000,1111},$ $\psi_{0001,1110}\}$, they are locally equivalent to the Bell basis in the $2 \otimes 2$ subspace  $\text{span}\{|000\rangle, |111\rangle\}_{\mathrm{ABC}} \otimes \text{span}\{|0\rangle, |1\rangle\}_{\mathrm{D}}$ $\subset (\mathbb{C}^2)^{\otimes 3} \otimes \mathbb{C}^2 $ through bipartition ABC$|$D. \ Hence, $\mathcal{N}_7^{(4)}$ is locally indistinguishable in this bipartition, which is shown at the $[\psi_{0001, 1110}, \ \psi_{0000,1111}]$ entry of Table \ref{tb1};
\smallskip
\\ \indent(2) For the triples $\{\psi_{0000,1111}, \ \psi_{0010}\}$ and $\{\psi_{0001,1110},$ $\psi_{0010}\}$, they are locally equivalent to three Bell states in subspace $\text{span}\{|000\rangle, |111\rangle\}_{\mathrm{ABD}}  \otimes \text{span}\{|0\rangle, |1\rangle\}_{\mathrm{C}} \subset (\mathbb{C}^2)^{\otimes 3} \otimes \mathbb{C}^2$ through bipartition ABD$|$C,  and subspace $\text{span}\{|00\rangle, |11\rangle\}_{\mathrm{AB}} \otimes \text{span}\{|01\rangle, |10\rangle\}_{\mathrm{CD}} \subset (\mathbb{C}^2)^{\otimes 2} \otimes (\mathbb{C}^2)^{\otimes 2} $ through AB$|$CD respectively. \ Therefore, $\mathcal{N}_7^{(4)}$ is locally indistinguishable through bipartitions ABD$|$C and AB$|$CD, \ which is shown by the third column of Table \ref{tb1};
\smallskip
\\ \indent(3) Similarly, for subsets $\{\psi_{0000,1111}, \psi_{0001,1110}, \psi_{0100}\}$ and $\{\psi_{0000,1111}, \psi_{0001,1110}, \psi_{1000}\}$, they are locally indistinguishable through bipartitions ACD$|$B, AC$|$BD and bipartitions BCD$|$A, BC$|$AD respectively, which is shown by the fourth and the fifth column in Table \ref{tb1}.
Notice that the $2^3-1 = 7$ different bipartitions are just right filled into the $2 \times 4$ enties except the ``null'' one. Therefore, $\mathcal{N}_7^{(4)} $ is genuinely nonlocal.
}
\end{Exm}

\medskip

\noindent {(ii)} In five-qubit system $(\mathbb{C}^2)^{\otimes 5}$, we also have the five-qubit GHZ basis:
$$|\phi_{0,31}\rangle, \ |\phi_{1,30}\rangle, \ \cdots \ , \ |\phi_{15,16}\rangle$$
which is now reindexed as
\begin{equation}\begin{split}\label{GHZ32}
& |\phi_{00000,11111}\rangle = \frac{|00000\rangle \pm |11111\rangle}{\sqrt{2}}, \\
& |\phi_{00001,11110}\rangle = \frac{|00001\rangle \pm |11110\rangle}{\sqrt{2}}, \\
& \phantom{|\phi_{00001,11110}\rangle \ }  \vdots    \\
& |\phi_{01111,10000}\rangle = \frac{|01111\rangle \pm |10000\rangle}{\sqrt{2}}.  \\
\end{split}\end{equation}
Similarly as $\mathcal{N}_7^{(4)}$, the $11$-ary subset
\begin{equation*}\begin{split}
\mathcal{N}_{11}^{(5)} = \{&\phi_{00000,11111}, \ \phi_{00001,11110}, \ \phi_{00010}, \ \phi_{00100},  \\
& \phi_{01000}, \ \phi_{10000}, \ \phi_{00110}, \ \phi_{01010}, \ \phi_{10010} \ \}
\end{split}\end{equation*}
of the GHZ basis (\ref{GHZ32}) can also be checked to be genuinely nonlocal, with aid of Table \ref{tb2}. \ For example, for the subset $\{\phi_{00000,11111},$ $\phi_{00001, 11110}, \ \phi_{00010}\}$, it is locally indistinguishable in the ABCE$|$D, ABC$|$DE bipartitions, as shown by the third column of Table \ref{tb2}. \ One can go through Table \ref{tb2} routinly and find that the $2^{4} - 1 = 15$ bipartitions are right filled into the $2 \times 8$ entries except the ``null'' one. Thus, $\mathcal{N}_{11}^{(5)}$ is genuinely nonlocal.

\bigskip

It's easy to generalize Example \ref{exm1} to the more-partite cases. For example, in $(\mathbb{C}^2)^{\otimes 6}$, one can construct
\begin{equation*}\begin{split}
\mathcal{S}_{19}^{(6)} = \{&\varphi_{000000,111111}, \ \varphi_{000001,111110}, \ \varphi_{000010}, \ \varphi_{000100},\\
& \varphi_{001000}, \ \ \varphi_{010000}, \ \ \varphi_{100000}, \ \ \varphi_{000110}, \ \ \varphi_{001010},\\
& \varphi_{010010}, \ \ \varphi_{100010}, \ \ \varphi_{001100}, \ \ \varphi_{010100}, \ \ \varphi_{100100}, \\
& \varphi_{011000}, \ \ \varphi_{101000}, \ \ \varphi_{110000} \}
\end{split}\end{equation*}
which also contains 2 conjugate pairs. However, the next example shows that 19 is by no means the smallest cardinality.

\bigskip

\begin{Exm}\label{exm2}
\textnormal{
In six-qubit system $(\mathbb{C}^2)^{\otimes 6}$, the $15$-ary subset
\begin{equation*}\begin{split}
\mathcal{N}_{15}^{(6)} = \{&\varphi_{000000,111111}, \ \varphi_{000001,111110}, \ \varphi_{000010,111101}, \\
& \ \varphi_{000011,111100}, \ \ \varphi_{000100}, \ \ \varphi_{001000}, \ \ \varphi_{010000},\\
& \ \varphi_{100000}, \ \ \varphi_{001100}, \ \ \varphi_{010100}, \ \ \varphi_{100100} \ \}
\end{split}\end{equation*}
of the six-qubit GHZ basis is genuinely nonlocal. Again, one can check this fact with Table \ref{tb3}. Unlike $\mathcal{S}_{19}^{(6)}$, $\mathcal{N}_{15}^{(6)}$ contains 4 conjugate pairs and these conjugate pairs are indistinguishable through bipartitions ABCDE$|$F, ABCDF$|$E and ABCD$|$EF, as shown in the second column of Table \ref{tb3}. \ Here in this example, for convenience of description, we place ``{\small ABCDEF$|\emptyset$}'' into the [$\varphi_{000000, 111111}, \ \varphi_{000000, 111111}$] entry, meaning nothing other than ``the states are indistinguishable in none of the bipartitions''. \ Notice that from this column, one can generate the columns behind by ``moving'' some of the parties from the left partition to the right partition. For example, for the ``$\varphi_{000100}$'' (third) column, just move ``D'' from the left side to the right side and for the ``$\varphi_{001100}$'' (seventh) column, move parties ``CD'' and so on. We will characterize in the proof of Theorem \ref{thm1} that this corresponds to a certain ``cosets structure'' about the group of all bipartitions. \ In fact, the same pattern  also appears in Table \ref{tb1} and Table \ref{tb2}, once we
place ``{\small ABCD$|\emptyset$}'' and ``{\small ABCDE$|\emptyset$}'' to the ``{null}'' entries  \cite{zero}.
}
\end{Exm}

\bigskip

\begin{Exm}\label{exm3}
\textnormal{
Going back to the five-qubit case, one can check that the subset
\begin{equation*}\begin{split}
\mathcal{S}_{11}^{(5)} = \{&\phi_{00000,11111}, \ \phi_{00001,11110}, \ \phi_{00010,11101}, \\
&\phi_{00011,11100}, \ \ \phi_{00100}, \ \ \phi_{01000}, \ \ \phi_{10000}\}
\end{split}\end{equation*}
of the five-qubit GHZ basis, which contains 4 conjugate pairs, is also genuinely nonlocal. But it has the same cardinality as $\mathcal{N}_{11}^{(5)}$. \ For the six-qubit case, we can further construct genuinely nonlocal subset
\begin{equation*}\begin{split}
\mathcal{N}_{19}^{(6)} = \{&\varphi_{000000,111111}, \ \varphi_{000001,111110}, \ \varphi_{000010,111101}, \\
& \varphi_{000100,111011}, \ \varphi_{000011,111100}, \ \varphi_{000110,111001}, \\
& \varphi_{000101,111010}, \ \ \varphi_{000111,111000}, \ \ \varphi_{001000},\\
& \varphi_{010000}, \ \ \varphi_{100000} \}
\end{split}\end{equation*}
that contains 8 conjugate pairs. We see however that it doesn't have smaller cardinality than $\mathcal{N}_{15}^{(6)}$.
}
\end{Exm}

\smallskip

The essence behind all the above examples will be revealed in the next subsection, using the language of group theory \cite{Rotman}.

\bigskip

\bigskip

\centerline{\bf 3. The unified characterization}

\bigskip

For each example in the last subsection, the $2^{N-1}$ bipartitions (including the ``null'' one \cite{zero}) for the $N$ subsystems actually form an abelian group $G'$. Each group $G'$ can be devided into bunches of cosets which correspond to columns of Table \ref{tb1} -- Table \ref{tb3}, provided a certain subgroup $H$ of $G'$ is specified. The subgroup $H$ defines the conjugate pairs in each genuinely nonlocal set. To be more specific, we present the following theorem for which the proof is constructive.

\smallskip

{\color{black}
\begin{Thm}\label{thm1}
\color{black}
\ In $N$-qubit system $(\mathbb{C}^2)^{\otimes N}$, genuinely nonlocal subsets $\mathcal{N}_t$ of the $N$-qubit GHZ basis
$$\{\xi_{0, 2^N-1}, \ \xi_{1, 2^N-2}, \ \cdots, \ \xi_{2^{N-1}-1, \ 2^{N-1}}\}$$
that have cardinality $|\mathcal{N}_t| = 2^t + 2^{N-t} - 1$ can be constructed, where $1 \leq t \leq N-1$.
\end{Thm}}

\begin{proof}
\ Denote $G$ the finite abelian group of all $N$-tuples, whose elements are $\mathbf{a} = (a^{(1)}, a^{(2)}, \cdots, a^{(N)}) \in \{0, 1\}^N$ and whose  additive operation ``$+$'' is defined as: \ $\mathbf{a} + \mathbf{b} = (a^{(1)} \oplus b^{(1)}, a^{(2)} \oplus b^{(2)}, \cdots, a^{(N)} \oplus b^{(N)})$, where ``$\oplus$'' is the ``mod 2'' addition. Obviously, $I = \{(00\cdots0), \ (11\cdots 1)\}$ is a subgroup of $G$ with order 2. Denote $G' = G/I$ its quotient group and the elements of $G'$ are denoted as $a = [\mathbf{a}] = [a^{(1)}, a^{(2)}, \cdots, a^{(N)}]$, signifying the coset $\mathbf{a} + I$ with representative element $\mathbf{a} \in G$. The additive operation of $G'$ is then induced as: \ $[\mathbf{a}] + [\mathbf{b}] = [a^{(1)} \oplus b^{(1)}, a^{(2)} \oplus b^{(2)}, \cdots, a^{(N)} \oplus b^{(N)}]$, which is well-defined. \ Notice that all elements of $G'$ have order 2 and therefore, $G'$ can also be regarded as an ``$\mathbb{F}_2$ -- linear space'', where $\mathbb{F}_2 = \{0, 1\}$ is regarded as a field. Since $|G'| = 2^{N-1}$, any proper subgroup $H < G'$ must have order $|H| = 2^{t-1} \ (1 \leq t \leq N-1)$. Moreover, $H = \langle h_1 \rangle$ {\footnotesize $\bigoplus$} $\cdots$ {\footnotesize $\bigoplus$} $\langle h_{t-1} \rangle$, where $\{h_1, \cdots, h_{t-1}\} \subset H$ is a set of basis that generates $H$ (regarded as a linear subspace of $G'$) and ``{\footnotesize $\bigoplus$}'' is the direct sum \cite{subgroup}. That is to say, any element $h \in H$ can be uniquely written as $h = m_1 h_1 + \cdots + m_{t-1} h_{t-1}$, where $m_i \in \mathbb{F}_2 \ (1 \leq i \leq t-1)$. \ Now, denote all the distinct cosets of $H$ as $g_1 + H, \cdots, g_{c} + H$, \ where $g_l + H \triangleq \{g_l + h: h \in H\}$ for $1 \leq l \leq c = 2^{N-t} - 1$. Since $G' = H \cup (g_1 + H) \cup \cdots \cup (g_c + H)$, \ any element $g \in G'$ must be located in $H$ or in certain $g_l + H$. Namely, $g = h$ or $g = g_l + h$ for some $h \in H$ and certain $l \in \{1, \cdots c\}$.

Given the above proper subgroup $H < G'$, we now construct the genuinely nonlocal subsets $\mathcal{N}_t$ of the $N$-qubit GHZ basis. For any element $h \in H$, denote $\mathbf{h}, \ \overline{\mathbf{h}} \in G$ the pair of $N$-tuples such that [$\mathbf{h}] = [\overline{\mathbf{h}}]  = h$. Obviously, $\overline{\mathbf{h}} = \mathbf{h} + (11\cdots 1)$. Also, for the aforemention $g_l \in G' \ (1 \leq l \leq c)$, \ denote $\mathbf{g}_l = (g_l^{(1)}, g_l^{(2)}, \cdots, g_l^{(N)})\in G$ one of the $N$-tuples such that $[\mathbf{g}_l] = g_l$. Let
$$\mathcal{N}_t = \{\xi_{\mathbf{h}, \overline{\mathbf{h}}} \ | \ h \in H\} \cup \{\xi_{\mathbf{g}_1}, \cdots, \xi_{\mathbf{g}_c}\},$$
then $|\mathcal{N}_t| = 2^t + 2^{N-t} - 1$ and we next prove that it is genuinely nonlocal. For any bipartition $S|\overline{S}$ $(\emptyset \subsetneq S \subsetneq \{1, 2, \cdots, N\})$, denote $\mathbf{p}_S = (p^{(1)}, p^{(2)}, \cdots, p^{(N)})$, where
\begin{equation*}
p^{(n)} = \left\{
\begin{array}{ll}
1, & \ n \in S  \\
0, & \ n \notin S \\
\end{array} \right.
\end{equation*}
and $\overline{\mathbf{p}}_S = \mathbf{p}_S + (11\cdots 1) = \mathbf{p}_{\overline{S}}$. By discussions in the last paragraph, for $[\mathbf{p}_S ]= [\overline{\mathbf{p}}_S ] \triangleq p_{S|\overline{S}} \in G'$, there must exist some $h_{S|\overline{S}} \in H$ such that $$p_{S|\overline{S}} = h_{S|\overline{S}} \ ,$$  or $$p_{S|\overline{S}} = g_l + h_{S|\overline{S}}$$  for certain $l \in \{ 1, \cdots, c\}$. \ For the former case, obviously, the two conjugate pairs $\xi_{\mathbf{00\cdots 0}, \ \mathbf{11\cdots 1}}$ and $\xi_{\mathbf{h}_S, \ \overline{\mathbf{h}}_S}$ are locally equivalent to the Bell basis in the $S|\overline{S}$ bipartition. \ For the latter case, across the bipartition $S|\overline{S}$, \ the triple $\{\xi_{\mathbf{g}_l}, \ \xi_{\mathbf{h}_S, \overline{\mathbf{h}}_S}\}$ is locally equivalent to three Bell states in the \ $2 \otimes 2$ \ subspace $\mathcal{H}^{S}_2 \otimes \mathcal{H}^{\overline{S}}_2 \ \subset \ (\mathbb{C}^2)^{\otimes |S|} \otimes (\mathbb{C}^2)^{\otimes |\overline{S}|} $, \ where
$$\mathcal{H}^{S}_2 =  \text{span}\left\{\left.\bigotimes_{n \in S}\left|g_l^{(n)} \oplus j\right> \right| j = 0, 1 \right\}$$
and
$$\mathcal{H}^{\overline{S}}_2 =  \text{span}\left\{\left.\bigotimes_{n \in \overline{S}}\left|g_l^{(n)} \oplus j\right> \right|  j = 0, 1 \right\}.$$
As a result, $\mathcal{N}_t$ is indistinguishable in all bipartitions and so it is genuinely nonlocal.
\end{proof}

\medskip
Immediately, we achieve genuinely nonlocal subset of the $N$-qubit GHZ basis with cardinality $\Theta(\sqrt{2^N})$, on condition that $N$ is large:
{\color{black}
\begin{Cor}\label{thm2}
\color{black}
\ In $N$-qubit system $(\mathbb{C}^2)^{\otimes N}$, \ $m_2(N)$ states among the $N$-qubit GHZ basis suffice to exhibit genuine nonlocality, where
\begin{equation*}
m_2(N) = \left\{
\begin{array}{ll}
2^{M+1} - 1, &  \ \ \ N  = 2 M \\
\\
2^{M+\log_2 3} - 1, & \ \ \ N = 2M + 1. \\
\end{array} \right.
\end{equation*}
\end{Cor}
}

\smallskip

\begin{proof}
\ In either case, $|\mathcal{N}_t| = 2^t + 2^{N-t} - 1$ achieve its minimum when $t = M$.
\end{proof}

\medskip

Along the same vein, one can generalize the above result to the $N$-qudit cases, wherever the local dimension $d$ are even integers. For cases where $d$ is odd, we leave it open.

\begin{Thm}\label{prop3}
\ In $N$-qudit system $(\mathbb{C}^d)^{\otimes N}$ where $d$ is even, $m_d(N)$ states among the canonical generalized GHZ basis (\ref{canonical}) suffice to exhibit genuine nonlocality, where
\begin{equation*}\begin{split}
& \ m_d(N) =  \\
& \ \ \  \ \ \ \left\{\begin{array}{ll}
2^M (1 + \frac{d}{2^{\lfloor \log_2d \rfloor}}) - 1, & \ \ N + \lfloor \log_2 \frac d2 \rfloor = 2 M \\
\\
2^M (2 + \frac{d}{2^{\lfloor \log_2d \rfloor}}) - 1, & \ \ N + \lfloor \log_2 \frac d2 \rfloor = 2M + 1. \\
\end{array} \right.
\end{split}\end{equation*}
\end{Thm}

\medskip

The proof of Theorem \ref{prop3} is given in the Appendix. Note that when $N$ approaches infinity, $m_d(N) =$ {$\Theta(\sqrt{2^N})$}. That is, a proportion {\small $\Theta[1/(\frac{d}{\sqrt{2}})^{N}]$} of the canonical generalized GHZ basis (\ref{canonical}) suffice to exhibit genuine nonlocality, whenever $d$ is even and $N$ is large.

\smallskip

\medskip
\section{$d + 1$ genuinely nonlocal generalized GHZ states in $(\mathbb{C}^d)^{\otimes N}$}\label{section4}
\medskip

The results of Section \ref{section3} indicate that when a fixed local dimension $d$ and a large number of parties $N$ are considered, a small proportion of the canonical generalized GHZ basis suffices to exhibit genuine nonlocality in $(\mathbb{C}^d)^{\otimes N}$. \ However, if we instead consider a fixed number of parties $N$, and allow $d$ to grow arbitrarily, then a similar argument doesn't hold anymore. We are handling these situations here in this section. \ As we will show, given any number of parties $N$, when the local dimension $d$ is sufficiently large and when the generalized GHZ states considered are not confined to the canonical form, $d + 1$ genuinely nonlocal generalized GHZ states can always be constructed in system $(\mathbb{C}^d)^{\otimes N}$.

Looking back to the genuinely nonlocal sets constructed in the last section, for each bipartition, there is always (at least) $d + 1$ states being locally equivalent to $d+1$ maximally entangled states in a certain $d \otimes d$ product subspace $\mathcal{H}_{d \otimes d} \simeq \mathbb{C}^d \otimes \mathbb{C}^d$. \ Here, a more cunning idea is: to make $d+1$ states locally indistinguishable, the $(d + 1)$-th state needn't be supported on the same $d \otimes d$ subspace. To be more specific, consider the three orthogonal states
\begin{equation}\begin{split}\label{alpha}
& |\alpha_1\rangle = \frac{|00\rangle + |11\rangle}{\sqrt{2}}, \ \ |\alpha_2\rangle = \frac{|00\rangle - |11\rangle}{\sqrt{2}}, \\
& |\alpha_3\rangle = |01\rangle
\end{split}\end{equation}
in system $\mathbb{C}^2 \otimes \mathbb{C}^2$, which have been shown to be locally indistinguishable by Ghosh et al. \cite{Ghosh02}. Ghosh et al. proved this fact by calculating an upper bound on the distillable entanglement of an ingeniously constructed four-party state, \ which will then induce a contradiction provided local distinguishability assumption were made. Actually, there are more concise ways to proof local indistinguishability of these states: for example, it's routine to prove that they are locally irreducible, \ using the TOPLM technique as in \cite{Walgate02}. Intuitively but informally, the local indistinguishability of these states can be understood in such a way: To distinguish $|\alpha_1\rangle$ and $|\alpha_2\rangle$, we have to measure both subsystems along the $\{|+\rangle, |-\rangle\}$ basis, by which we get ``$++$'' or ``$--$'' if we have $|\alpha_0\rangle$, and get ``$+-$'' or ``$-+$'' if we have $|\alpha_1\rangle$; \ Simultaneously, we must also distinguish $|\alpha_3\rangle$ but unfortunately, $|\alpha_3\rangle = |01\rangle$ has nonzero overlap with each one of $|++\rangle, |+-\rangle, |-+\rangle, |--\rangle$. Therefore, the states are locally indistinguishable. In fact, the assumption that $|\alpha_3\rangle$ has nonzero overlap with just one among $|++\rangle, |+-\rangle, |-+\rangle, |--\rangle$ (equivalently, nonzero overlap with span$\{|++\rangle, |+-\rangle, |-+\rangle, |--\rangle\}$) suffices to induce indistinguishability of the set.

\begin{figure}[t]
\centering
\includegraphics[scale=0.33]{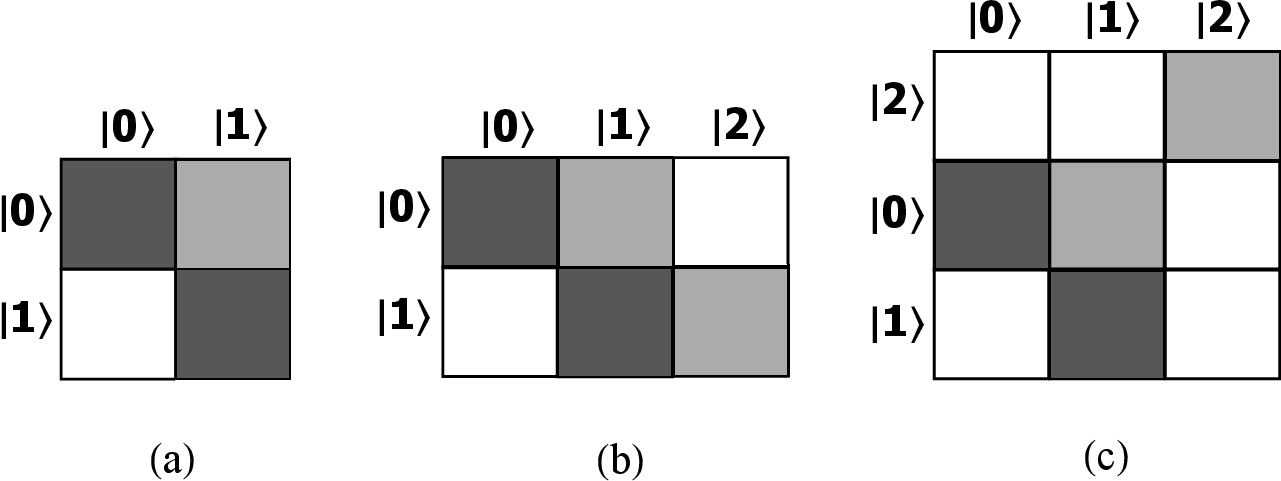}
\captionsetup{justification=raggedright}
\caption
{\small (a) States $\{|\alpha_1\rangle, |\alpha_2\rangle, |\alpha_3\rangle\}$ in $\mathbb{C}^2 \otimes \mathbb{C}^2$; \ (b) States $\{|\beta_1\rangle, |\beta_2\rangle, |\beta_3\rangle\}$ in $\mathbb{C}^2 \otimes \mathbb{C}^3$; \ \ (c) States $\{|\gamma_1\rangle, |\gamma_2\rangle, |\gamma_3\rangle\}$ in $\mathbb{C}^3 \otimes \mathbb{C}^3$. \ The dark tiles signify the conjugate pairs and the grey ones signify the third states.
}\label{fig1}
\end{figure}

Inspired by this naturally, one might further consider local distinguishability of the three states
\begin{equation}\begin{split}
& |\beta_1\rangle = \frac{|00\rangle + |11\rangle}{\sqrt{2}}, \ \ |\beta_2\rangle = \frac{|00\rangle - |11\rangle}{\sqrt{2}}, \\
& |\beta_3\rangle = \frac{|01\rangle + |12\rangle}{\sqrt{2}}
\end{split}\end{equation}
in $\mathbb{C}^2 \otimes \mathbb{C}^3$ or alternatively, states
\begin{equation}\begin{split}\label{eta}
& |\gamma_1\rangle = \frac{|00\rangle + |11\rangle}{\sqrt{2}}, \ \ |\gamma_2\rangle = \frac{|00\rangle - |11\rangle}{\sqrt{2}}, \\
& |\gamma_3\rangle = \frac{|01\rangle + |22\rangle}{\sqrt{2}}
\end{split}\end{equation}
in $\mathbb{C}^3 \otimes \mathbb{C}^3$. Similar as $\{|\alpha_1\rangle, |\alpha_2\rangle, |\alpha_3\rangle\}$, the last states in both sets have nonzero overlap with $\text{span}\{|++\rangle,$ $|+-\rangle, |-+\rangle, |--\rangle\} = $ span$\{|0\rangle, |1\rangle\} \otimes \text{span}\{|0\rangle, |1\rangle\}$ (see FIG \ref{fig1}). Beyond questions, it's reasonable to speculate that they are also locally indistinguishable. Unlike $\{|\alpha_1\rangle, |\alpha_2\rangle, |\alpha_3\rangle\}$ however, we failed to prove this rigorously with the methods from both \cite{Walgate02} and \cite{Ghosh02}. Fortunately, we instead find that all such sets are PPT-indistinguishable indeed. For convenience of explanation, we first begin with $\{|\beta_1\rangle, |\beta_2\rangle, |\beta_3\rangle\}$. There is a same argument for its $d$-dimensional analogue:

\begin{Lem}\label{lemma2}
\ In two-partite system $\mathbb{C}^d \otimes \mathbb{C}^{2d - 1}$ $(d \geq 2)$, the following $d + 1$ orthogonal states:
\begin{equation}\begin{split}
& \ |\chi^{(d)}_{k}\rangle = \frac{1}{\sqrt{d}} \sum_{j = 0}^{d-1} \omega_{d}^{j(k-1)} \ |jj\rangle, \ \ \ (1 \leq k \leq d, \ \ \omega_{d} = e^{\frac{2\pi i}{d}})\\
& and \\
& \ |\chi^{(d)}_{d+1}\rangle = \frac{1}{\sqrt{d}} \sum_{j = 0}^{d-1}  \ |j, \ j+d-1 \rangle
\end{split}\end{equation}
are PPT-indistinguishable.
\end{Lem}

\medskip

\begin{proof}
For the last state $|\chi^{(d)}_{d+1}\rangle$, its component ``$|0, d-1\rangle$'' has nonzero overlap with the $d \otimes d$ product subspace
$$\text{span}\{|0\rangle, \cdots, |d-1\rangle\} \otimes \text{span}\{|0\rangle, \cdots, |d-1\rangle\}.$$
Hence, the lemma is a special case of Lemma \ref{lemma3}, for which a more general proof will be given.
\end{proof}

\bigskip

With Lemma \ref{lemma2}, we can now construct $d + 1$ genuinely nonlocal generalized GHZ states in system $\mathbb{C}^d \otimes \mathbb{C}^d \otimes \mathbb{C}^d$.

\medskip

\begin{Thm}\label{thm3}
\ In three-partite system $\mathbb{C}^d \otimes \mathbb{C}^d \otimes \mathbb{C}^d$ where $d \geq 4$, there exist $d + 1$ orthogonal generalized GHZ states which are genuinely nonlocal.
\end{Thm}

\medskip

\begin{proof}
For $d = 4$, we present 5 orthogonal generalized GHZ states in $\mathbb{C}^4 \otimes \mathbb{C}^4 \otimes \mathbb{C}^4$ as
\begin{equation}\begin{split}
& |\eta_1\rangle = \frac{|000\rangle + |111\rangle + |222\rangle + |333\rangle}{2}, \\
& |\eta_2\rangle = \frac{|000\rangle + i|111\rangle - |222\rangle - i|333\rangle}{2}, \\
& |\eta_3\rangle = \frac{|000\rangle - |111\rangle + |222\rangle - |333\rangle}{2}, \\
& |\eta_4\rangle = \frac{|000\rangle - i|111\rangle - |222\rangle + i|333\rangle}{2}, \\
& |\eta_5\rangle = \frac{|011\rangle + |202\rangle + |330\rangle + |123\rangle}{2}. \\
\end{split}\end{equation}
To prove their genuine nonlocality, first consider the bipartition AB$|$C: In the subspace $\mathcal{H}'_{\mathrm{AB}} \otimes \mathcal{H}_{\mathrm{C}}$, where $\mathcal{H}'_{\mathrm{AB}} = \text{span}\{|00\rangle, |11\rangle, |22\rangle, |33\rangle, |01\rangle, |20\rangle, |12\rangle\}_{\mathrm{AB}}$, the states are locally equivalent to $\{|\chi^{(4)}_{1}\rangle, \cdots, |\chi^{(4)}_{5}\rangle\}$ (case $d = 4$ for Lemma \ref{lemma2}) in the bipartite system $\mathbb{C}^4 \otimes \mathbb{C}^7$, so they are locally indistinguishable through AB$|$C. This fact is shown in FIG \ref{fig2}. \ For the other two bipartitions, the same argument can also be easily verified (for example, A$|$BC is also shown in FIG \ref{fig2}). Thus, the generalized GHZ states $\{|\eta_1\rangle, \cdots, |\eta_5\rangle\}$ are genuinely nonlocal.

Notice that the construction of the last state $|\eta_5\rangle$ is crucial for the above result: The first three components ``$|011\rangle$'', ``$|202\rangle$'' and ``$|330\rangle$'' of $|\eta_5\rangle$ are intentionally arranged in such a way that they have nonzero overlap with the $4 \otimes 4$ subspace
$$\text{span}\{|0\rangle, |1\rangle, |2\rangle, |3\rangle\} \otimes \text{span}\{|00\rangle, |11\rangle, |22\rangle, |33\rangle\},$$
in bipartitions A$|$BC, B$|$CA and C$|$AB respectively; The last component ``$|123\rangle$'' is just to ensure $|\eta_5\rangle$ to be a generalized GHZ state and to be orthogonal to others. With these, Lemma \ref{lemma2} can then be applied for each bipartition. Notably, such constructions are always possible for $d \geq 4$. Taking $d = 5$ for example, the last state can be constructed as
$$\frac{|011\rangle + |202\rangle + |334\rangle + |143\rangle + |420\rangle}{\sqrt{5}},$$
which is locally equivalent to $|\chi^{(5)}_{6}\rangle$ (the $d = 5$ case for Lemma \ref{lemma2}), in each of the three bipartitions. For $d = 6$, the last state can be constructed as
$$\frac{|011\rangle + |202\rangle + |334\rangle + |145\rangle + |450\rangle + |523\rangle}{\sqrt{6}}$$
and so forth for any larger $d$.
\end{proof}

\begin{figure}[t]
\centering
\includegraphics[scale=0.30]{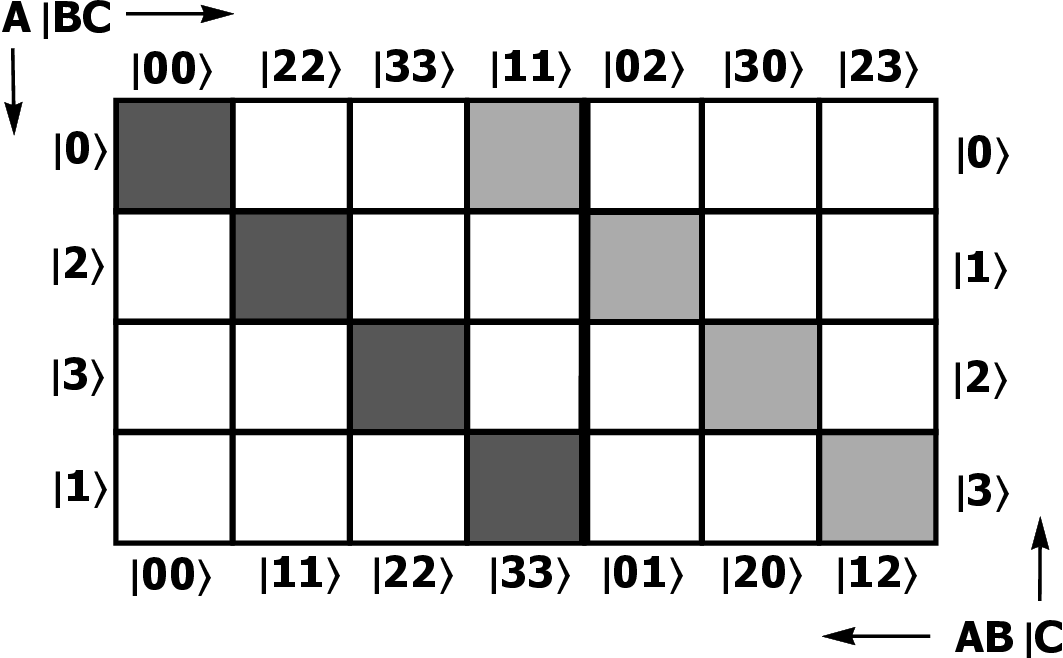}
\captionsetup{justification=raggedright}
\caption
{\small \ \ States $\{|\eta_1\rangle, \cdots, |\eta_5\rangle\}$ in bipartitions A$|$BC and AB$|$C. The dark tiles signify the four states $|\eta_1\rangle, \cdots,$ $|\eta_4\rangle$ and the grey ones signify $|\eta_5\rangle$. \ \ In both bipartitions, they are locally equivalent to $\{|\chi^{(4)}_{1}\rangle, \cdots, |\chi^{(4)}_{5}\rangle\}$ in system $\mathbb{C}^4 \otimes$ $\mathbb{C}^7$ (case $d = 4$ of Lemma \ref{lemma2}).
}\label{fig2}
\end{figure}

\medskip
To construct $d+1$ genuinely nonlocal generalized GHZ states in more-partite systems, we need a slightly more generalized version of Lemma \ref{lemma2}:

\begin{Lem}\label{lemma3}
\ \ In two-partite system $\mathbb{C}^{2d-1} \otimes \mathbb{C}^{2d-1}$ ($d \geq 2$), the following $d + 1$ orthogonal states are PPT-indistinguishable:
\begin{equation}\begin{split}\label{d+1}
& \ \ |\lambda_{k}\rangle = \frac{1}{\sqrt{d}} \sum_{j = 0}^{d-1} \omega_{d}^{j(k-1)} \ |jj\rangle, \ \ \ (1 \leq k \leq d, \ \ \omega_{d} = e^{\frac{2\pi i}{d}})\\
& and \\
& \ \ |\lambda_{d+1}\rangle = \frac{1}{\sqrt{d}} \sum_{j = 0}^{d-1}  \ |e_j f_j\rangle,
\end{split}\end{equation}
where the distinct $e_j$'s and the distinct $f_j$'s satisfy:
\begin{enumerate}
\item[\textnormal{(i) \ }] $0 \leq e_j, f_j \leq 2d-2$;
\item[\textnormal{(ii) \ }] There is at least one $j_0 \in \{0, \cdots, d-1\}$ such that $e_{j_0} \leq d-1$ and $f_{j_0} \leq d-1$;
\item[\textnormal{(iii) \ }] For those $j$ such that $e_j, f_j \ \leq \ d-1$, there must be $e_j \neq f_j$.
\end{enumerate}
\end{Lem}

\smallskip

For readability, the proof of this lemma is placed to the Appendix. Notice that we can always restrict our discussion on a minimal subspace $\mathbb{C}^{m} \otimes \mathbb{C}^{n}$ $(d \leq m, n \leq 2d - 1)$ on which $|\lambda_{1}\rangle, \cdots, |\lambda_{d+1}\rangle$ are supported (just as Lemma \ref{lemma2}). This obviously won't affect the correctness. 

\smallskip
Now we take the four-partite case for example. In the four-qudit system $(\mathbb{C}^{11})^{\otimes 4}$, where the standard orthogonal basis for each subsystem is denoted as $\{|0\rangle, \cdots, |9\rangle, |10\rangle \triangleq |z\rangle\}$, the following 12 orthogonal generalized GHZ states:
\begin{equation*}\begin{split}
\ \ |\delta_{k}\rangle & = \frac{1}{\sqrt{11}} \sum_{j = 0}^{10} \omega_{11}^{j(k-1)} \ |jjjj\rangle, \\
& \phantom{ = \frac{1}{\sqrt{11}} \sum_{j = 0}^{10} \delta_{11}^{j(k-1)}|jjjj\rangle} (1 \leq k \leq 11, \ \ \omega_{11} = e^{\frac{2\pi i}{11}}) \\
\end{split}\end{equation*}
and
\begin{equation}\begin{split}
\ |\delta_{12}\rangle & = \frac{1} {\sqrt{11}} (|0111\rangle + |2022\rangle + |3303\rangle + |4440\rangle \ \ \ \ \ \  \\
\ & \phantom{\frac{1} {\sqrt{11}} \ } + |5566\rangle + |7878\rangle + |9zz9\rangle + |1234\rangle \ \ \ \ \ \  \\
\ & \phantom{\frac{1} {\sqrt{11}} \ \ \ } + |678z\rangle + |8695\rangle + |z957\rangle) \ \ \ \ \ \ \\
\end{split}\end{equation}
are locally equivalent to the 12 indistinguishable states $\{|\chi^{(11)}_{1}\rangle, \cdots, |\chi^{(11)}_{12}\rangle\}$ indicated by Lemma \ref{lemma2}, \ in each of the ``$1-3$'' bipartitions. For the ``$2-2$'' bipartitions, take AB$|$CD as an example:  In FIG \ref{fig3}, \ it is shown that $\{|\delta_1\rangle, \ \cdots, \ |\delta_{12}\rangle\}$ are locally equivalent to the following orthogonal states in bipartite system $\mathbb{C}^{19} \otimes \mathbb{C}^{19}$:
\begin{equation*}\begin{split}
& |\nu_{k}\rangle = \frac{1}{\sqrt{11}} \sum_{j = 0}^{10} \omega_{11}^{j(k-1)} \ |jj\rangle, \ \ \ (1 \leq k \leq 11, \ \ \omega_{11} = e^{\frac{2\pi i}{11}})\\
\end{split}\end{equation*}
and
\begin{equation}\begin{split}
& |\nu_{12}\rangle = \frac{1}{\sqrt{11}} (|a1\rangle + |b2\rangle + |3a\rangle + |4b\rangle + |56\rangle + |cc\rangle + |dd\rangle\\
& \phantom{|\nu_{12}\rangle = } + |ee\rangle + |ff\rangle + |gg\rangle + |hh\rangle),\\
\end{split}\end{equation}
where $\{|0\rangle, \cdots, |9\rangle, |10\rangle \triangleq |z\rangle, |11\rangle \triangleq |a\rangle, |12\rangle \triangleq |b\rangle, \cdots,$ $|18\rangle \triangleq |h\rangle\}$ is the standard orthogonal basis for both subsystems. One can check routinely that $\{|\delta_1\rangle, \cdots, |\delta_{12}\rangle\}$ are also locally equivalent to $\{|\nu_1\rangle, \cdots, |\nu_{12}\rangle\}$ in  bipartitions AC$|$BD and AD$|$BC (under some permutation of the basis $\{|0\rangle, \cdots, |h\rangle\}$). By Lemma \ref{lemma3} however, $\{|\nu_1\rangle, \cdots, |\nu_{12}\rangle\}$ are locally indistinguishable. Therefore, the above four-partite generalized GHZ states in $(\mathbb{C}^{11})^{\otimes 4}$ are genuinely nonlocal.

\begin{figure}[t]
\centering
\includegraphics[scale=0.27]{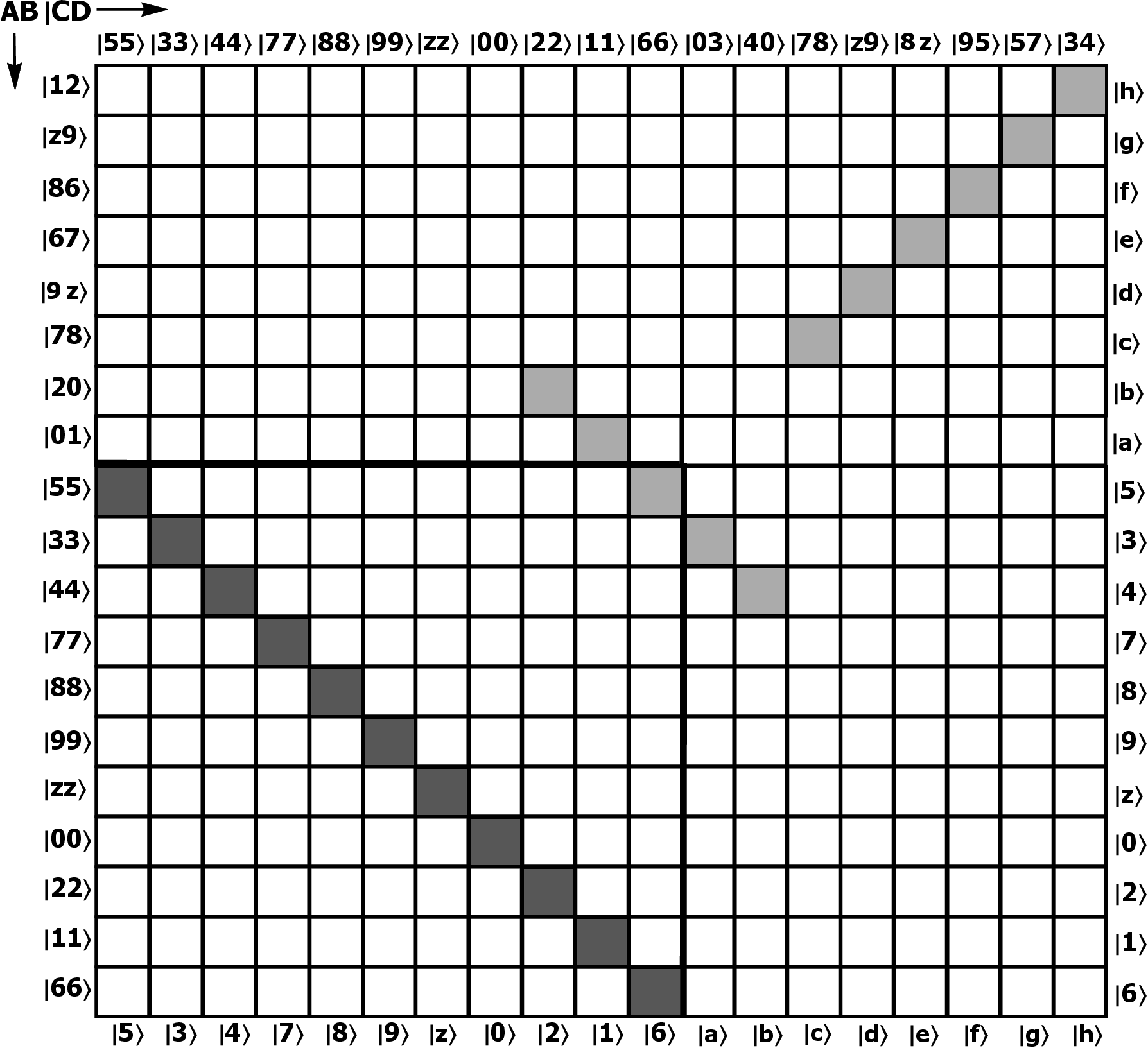}
\captionsetup{justification=raggedright}
\caption
{\small \ \ States $\{|\delta_1\rangle, \cdots, |\delta_{12}\rangle\}$ in the AB$|$CD bipartition. The dark tiles signify the eleven states \ $|\delta_1\rangle, \cdots,$ $|\delta_{11}\rangle$ \ and the grey ones represent $|\delta_{12}\rangle$. \ Such states are supported on $\mathbb{C}^{19} \otimes \mathbb{C}^{19}$ and they are locally equivalent to $\{|\nu_{1}\rangle, \cdots, |\nu_{12}\rangle\}$.
}\label{fig3}
\end{figure}

Similarly as the three-partite case, the components ``$|0111\rangle$'', ``$|2022\rangle$'', ``$|3303\rangle$'', ``$|4440\rangle$'' and ``$|5566\rangle$'', ``$|7878\rangle$'', ``$|9zz9\rangle$'' of the last state $|\delta_{12}\rangle$ here are intentionally arranged in such a way that they have nonzero overlap with the $11 \otimes 11$ subspaces
$$\text{span}\{|0\rangle, \cdots, |z\rangle\} \otimes \text{span}\{|000\rangle, \cdots, |zzz\rangle\}$$
and
$$\text{span}\{|00\rangle, \cdots, |zz\rangle\} \otimes \text{span}\{|00\rangle, \cdots, |zz\rangle\},$$
in the four ``$1-3$'' bipartitions and the three ``$2-2$'' bipartitions respectively. The remaining components are just to ensure $|\delta_{12}\rangle$ to be a generalized GHZ state and to be orthogonal to others. Also, such constructions can always be done when $d \geq 11$ and therefore we get a four-partite version of Theorem \ref{thm3}. \ Going a step further, for more general $N$-partite systems $(\mathbb{C}^d)^{\otimes N}$, similar constructions for the $(d+1)$-th generalized GHZ states are always possible as what we showed above, provided the local dimension $d$ is large enough. Therefore, by Lemma \ref{lemma3}, we have the following argument:
\begin{Thm}\label{thm4}
\ \ Given any number of parties $N$, there always exist $d + 1$ orthogonal generalized GHZ states in $(\mathbb{C}^d)^{\otimes N}$ which are genuinely nonlocal, in case that the local dimension $d$ is sufficiently large. 
\end{Thm}

\smallskip

{\color{black}
Normally, although a large local dimension $d$ is required here to achieve the cardinality ``$d+1$'', \ for systems with smaller $d$, by taking advantage of Lemma \ref{lemma3}, genuinely nonlocal sets with modest cardinality can still be similarly constructed. \ Taking $N=4$ for example,  although $d \geq 11$ is required for cardinality ``$d+1$'' as we mentioned, \ for the sightly more compromised ``$d+2$'' (or ``$d+3$'', ``$d+4$''), \ one can easily observe that only local dimension $d \geq 6$ (resp. $d \geq 4$, \ $d \geq 3$) suffices.

The results above are of exceptional interest when considering the so called ``trivial orthogonality--preserving local measurements'' (TOPLM) technique. In recent literature, TOPLM has been frequently utilized as a technique for detecting genuinely nonlocal sets \cite{Halder19,Zhang19, Yuan20, Shi20, Wang21, Shi2022, Shi22, Zhou22, Zhou23, Li23}. In fact, genuinely nonlocal sets detected (constructed) through this way are referred to as ``strongestly nonlocal sets'' by some researchers \cite{strongest nonlocality}.  Also, whatever called ``strongly nonlocal sets'' are all constructed this way by far. Not long ago, it has been proved in \cite{Li23} that for any system $(\mathbb{C}^d)^{\otimes N}$, genuinely nonlocal sets constructed through TOPLM must have cardinality no smaller than $d^{N-1} + 1$. Here however, Theorem \ref{thm4} show the existence of $d+1$ genuinely nonlocal states  in $(\mathbb{C}^d)^{\otimes N}$, provided $d$ is sufficiently large. Needless to say, in such asymptotic situation, the disparity between $d^{N-1} + 1$ and $d+1$ is huge. \ On the other hand, even beyond the asymptotic sense, for those smallest systems, the contrasts are already conspicuous: For $(\mathbb{C}^4)^{\otimes 3}$, genuinely nonlocal set of $5$ states exists by Theorem \ref{thm3}, while only sets of $17$ can be detected with TOPLM at best; \ For $(\mathbb{C}^3)^{\otimes 4}$, genuinely nonlocal set of $7$ exists (as indicated by discussions in the last paragraph) while only sets of $28$ can be detected with TOPLM at best. Undoubtedly, such comparisons do illustrate an evident limitation of the TOPLM technique for detecting genuine nonlocality. With just TOPLM, we could possibly miss out quite a lot of interesting genuinely nonlocal sets, especially the small ones. In addition, here we also point out that the genuinely nonlocal sets constructed through TOPLM may be redundant. For example, in \cite{Wang21}, the following set of strongly nonlocal basis consisting of generalized  GHZ states in $(\mathbb{C}^4)^{\otimes 3}$ has been constructed through TOPLM: \ $\mathcal{S} = \bigcup_{i = 1}^{16} \mathcal{S}_i$ where
\small\begin{equation}\begin{split}\label{noncanonical}
& \mathcal{S}_1 = \left\{\left. |000\rangle + i^{k} |121\rangle + (-1)^{k} |212\rangle + (-i)^{k} |333\rangle \right| k \in \mathbb{Z}_3 \right\}, \\
& \mathcal{S}_2 = \left\{\left. |003\rangle + i^{k} |111\rangle + (-1)^{k} |222\rangle + (-i)^{k} |330\rangle \right| k \in \mathbb{Z}_3 \right\}, \\
& \mathcal{S}_3 = \left\{\left. |030\rangle + i^{k} |112\rangle + (-1)^{k} |221\rangle + (-i)^{k} |303\rangle \right| k \in \mathbb{Z}_3 \right\}, \\
& ~~~~~~~~~~~~~~~~~~~~~~~~\vdots \\
& \mathcal{S}_{16} = \left\{\left. |032\rangle + i^{k} |120\rangle + (-1)^{k} |213\rangle + (-i)^{k} |301\rangle \right| k \in \mathbb{Z}_3 \right\}. \phantom{\frac12_{\frac12}}\\
\end{split}\end{equation}\normalsize
Whereas by Lemma \ref{lemma3}, one can easily observe that only $6$ states among them suffice to exhibit genuine nonlocality. That is, in the sense of nontriviality, the other $58$ states among this set are redundant.}

\medskip

\section{Conclusions}\label{section5}
{\color{black}
In this work, we study the problem of constructing small genuinely nonlocal sets consisting of generalized GHZ states in multipartite systems. Specifically, we consider systems $(\mathbb{C}^d)^{\otimes N}$ in two asymptotic situations: (i) $d$ is fixed and $N$ is large; (ii) $N$ is fixed and $d$ is large. \ In the former case, considering the canonical generalized GHZ bases, we find that only a proportion $\Theta[1/(\frac{d}{\sqrt{2}})^N]$ of the states among such bases suffice to exhibit genuine nonlocality, whenever the local dimension $d$ are even. It's noteworthy that $\Theta[1/(\frac{d}{\sqrt{2}})^N]$ approaches zero rapidly as $N$ grows. As for the latter case, we show the existence of $d+1$ genuinely nonlocal generalized GHZ states in $(\mathbb{C}^d)^{\otimes N}$. Moreover, within and beyond the asymptotic sense, the latter result also indicates some evident limitations of the TOPLM technique for detecting small genuinely nonlocal sets. }

{\color{black} Recently, the notion ``strong nonlocality'' has been discussed a lot by many reseachers.  However, all existing examples of strongly nonlocal sets by far are constructed through the TOPLM indeed. Here, we conjecture that the $d+1$ genuinely nonlocal GHZ states in Theorem \ref{thm4} are strongly nonlocal. Or, to put it another way, we conjecture that the $d + 1$ locally indistinguishable states in Lemma \ref{lemma3} to be locally irreducible. Notably, since any three locally indistinguishable states are always locally irreducible \cite{Halder19}, the conjecture is true for $d = 2$ at least. It's reasonable to speculate that the same holds for $d > 2$, as the structure of the states are similar: the first $d$ ``conjugate states'' (up to Fourier transformation) attached with the last ``stopper'' state, which has nonzero overlap with the $d \otimes d$ subspace on which the ``conjugate states'' are supported.}

There are also other questions left to be considered. For exmaple, for the canonical generalized GHZ bases, do genuinely nonlocal subsets constructed in Corollary \ref{thm2} and Theorem \ref{prop3} have the minimal cardinality? For the multi-qubit case in particular, though we cannot provide a rigorous proof by far, we conjecture the cardinality $m_2(N)$ to be minimal. Besides, what is the situation for cases when $d$ is odd? {\color{black} Moreover, it's also interesting to consider other types of multipartite quantum states. With them, can other interesting small genuinely nonlocal sets be further constructed?}
\\
\\

\medskip

\centerline{\textbf{ACKNOWLEDGMENTS}}
\bigskip
{\color{black} We thank the anonymous referee for pointing out the recent work about local random authentication \cite{SBGhosh} to us.} The authors of this work are supported by the National Natural Science Foundation of China \ (with Grant No.62272492, No.12301020), the Guangdong Basic and Applied Basic Research Foundation (Grant No.2020B1515020050, No.2023A1515012074) and the Science and Technology Projects in Guangzhou (Grant No.2023A04J0027).
\\
\\

\medskip

\centerline{\textbf{APPENDIX}}

\bigskip

\centerline{\bf 1. Proof of Theorem \ref{prop3}}

\begin{proof}
\ Suppose that a proper subgroup $H < G'$ is given, with $|H| = 2^{t-1}$ $(1 \leq t \leq N)$ and $g_1, \cdots, g_c \in G'$ $(c = 2^{N-t} - 1)$ the representative elements of its cosets, as discussed in the proof of Lemma \ref{thm1}. For any $h \in H$, denote $\mathbf{h} = (h^{(1)}, h^{(2)}, \cdots, h^{(N)})$ (one of) the $N$-tuple such that $[\mathbf{h} ] = h$ and $\mathbf{g}_l = (g_l^{(1)}, g_l^{(2)}, \cdots, g_l^{(N)})$  such that $[\mathbf{g}_l ] = g_l$ $(1 \leq l \leq c)$. Let
\begin{equation*}
\mathcal{N}_t = \left\{\Xi_k^{(\mathbf{h})} \ | \ h \in H, \ 1 \leq k \leq d \right\} \cup \left\{\Xi^{\mathbf{g}_l} \ | \ 1 \leq l \leq c \right\}
\end{equation*}
where
\begin{equation*}\begin{split}
& \left|\Xi_k^{(\mathbf{h})}\right> =  \\
& ~~~~~\frac{1}{\sqrt{d}} \sum_{j=0}^{d-1} \omega_d^{j(k-1)} \left|j \oplus\frac{h^{(1)}d}{2}\right>  \otimes \cdots \otimes \left|j \oplus \frac{h^{(N)}d}{2}\right>~~~~~\\
& \phantom{~~~~~\frac{1}{\sqrt{d}} \sum_{j=0}^{d-1} \omega_d^{j(k-1)} \left|j \oplus\frac{h^{(1)}d}{2}\right>  \otimes \cdots \otimes ~~~} (\omega_d = e^{2\pi i/d}) \\
\end{split}\end{equation*}
and
\begin{equation*}\begin{split}
\left| \Xi^{\mathbf{g}_l} \right> = \frac{1}{\sqrt{d}} \sum_{j=0}^{d-1} \ \left|j \oplus\frac{g_l^{(1)}d}{2} \right> \otimes \cdots \otimes \left|j \oplus \frac{g_l^{(N)}d}{2}\right>. \ \ \ \ \ \ \\
\end{split}\end{equation*}
For any bipartition $S|\overline{S}$ where $\emptyset \subsetneq S \subsetneq \{1, 2, \cdots, N\}$, which is indicated by $p_{S|\overline{S}} \in G'$ (as described in the proof of Lemma \ref{thm1}), there must exist some $h_{S|\overline{S}} \in H$ such that $p_{S|\overline{S}} = h_{S|\overline{S}}$ or $p_{S|\overline{S}} = h_{S|\overline{S}} + g_l$ for certain $1 \leq l \leq c$. For the former case, the subset
$$\left\{\Xi_k^{(\mathbf{00\cdots 0})}, \ \Xi_k^{(\mathbf{h}_{S|\overline{S}})}: 1 \leq k \leq d\right\}$$
is locally equivalent to $2d$ maximally entangled states in the $d \otimes d$ subspace $\mathcal{H}^{S}_d \otimes \mathcal{H}^{\overline{S}}_d \subset (\mathbb{C}^d)^{\otimes |S|} \otimes (\mathbb{C}^d)^{\otimes |\overline{S}|}$ through $S|\overline{S}$, where $\mathcal{H}^{S}_d =$ span$\{|j\cdots j\rangle_S\}_{j=0}^{d-1}$ and $\mathcal{H}^{\overline{S}}_d =$ span$\{|j\cdots j\rangle_{\overline{S}}\}_{j=0}^{d-1}$; \ For the latter case, also in bipartition $S|\overline{S}$, the subset
$$\left\{\Xi^{\mathbf{g}_l}, \ \Xi_k^{(\mathbf{h}_{S|\overline{S}})}: 1 \leq k \leq d\right\}$$
is locally equivalent to $d + 1$ maximally entangled states in $d \otimes d$ subspace $\mathcal{H}^{S}_d \otimes \mathcal{H}^{\overline{S}}_d \subset (\mathbb{C}^d)^{\otimes |S|} \otimes (\mathbb{C}^d)^{\otimes |\overline{S}|} $, where
$$\mathcal{H}^{S}_d =  \text{span}\left\{\left.\bigotimes_{n \in S}\left|j \oplus\frac{h^{(n)}d}{2}\right> \right| 0 \leq j \leq d-1 \right\}$$
and
$$\mathcal{H}^{\overline{S}}_d =  \text{span}\left\{\left.\bigotimes_{n \in \overline{S}}\left|j \oplus\frac{h^{(n)}d}{2}\right> \right| 0 \leq j \leq d-1 \right\}.$$
Therefore, $\mathcal{N}_t$ is indistinguishable through bipartition $S|\overline{S}$ and so for the other bipartitions. Hence, $\mathcal{N}_t$ is genuinely nonlocal. Notice that $|\mathcal{N}_t| = 2^{t-1}d + 2^{N-t} - 1 = \lambda \cdot 2^{t + \lfloor \log_2 \frac d2\rfloor} + 2^{N-t} - 1$, where $1 \leq \lambda = d/2^{\lfloor \log_2 d\rfloor} < 2$. $|\mathcal{N}_t|$ achieve its minimum with $t_m = \lfloor \frac{N - \lfloor \log_2 \frac d2\rfloor}{2}\rfloor$ and $|\mathcal{N}_{t_m}|$ is just the cardinality $m_d(N)$ shown above.
\end{proof}
Notice that $t_m > 1$ whenever $N \geq \lfloor \log_2 d\rfloor + 3$, in which case $|\mathcal{N}_{t_m}|$ is smaller than $|\mathcal{N}_1| = d + 2^{N-1} - 1$, the cardinality given by Proposition \ref{1d}.
\\
\\

\medskip

\centerline{\bf 2. Proof of Lemma \ref{lemma3}}

\begin{proof}
\ Assume contrarily that $|\lambda_{1}\rangle, \cdots, |\lambda_{d+1}\rangle$ are PPT-distinguishable. Then by definition, there must exist PPT-POVMs $\{M_1, \cdots, M_{d+1}\}$ (which means $M_k,$ $M_k^{\mathrm{T_A}} \geq 0$ for all $1 \leq k \leq d+1$) such that
\begin{equation}\label{Mi}
M_1 + \cdots + M_{d+1} = \mathbb{I},
\end{equation}
and
\begin{equation}
\mathrm{Tr}(M_k \Lambda_k) = 1  \ \ \ \ \ \ (1 \leq k \leq d+1)
\end{equation}
where $\Lambda_k = | \lambda_k \rangle \langle \lambda_k |$'s are the corresponding density operators and $ \mathbb{I} $ is the identity on the whole space $\mathbb{C}^{2d-1} \otimes \mathbb{C}^{2d-1}$. Denote
\begin{equation*}
P = \sum_{i=0}^{d-1} |i\rangle\langle i| \ \otimes \ \sum_{j=0}^{d-1} |j\rangle\langle j|.
\end{equation*}
the projection onto product subspace $\mathcal{H}_{d \otimes d} = $ span$\{|0\rangle,$ $\cdots, |d-1\rangle\}$ $\otimes$ span$\{|0\rangle, \cdots, |d-1\rangle\}$. Denote further that $P^{\bot} = \mathbb{I} - P$ and $\widetilde{M}_k = P M_k P$. Since the states $|\lambda_{1}\rangle, \cdots, |\lambda_{d}\rangle$ are supported on $\mathcal{H}_{d \otimes d}$, we have $P \Lambda_k P = \Lambda_k$ and thus
\begin{equation*}\begin{split}
\mathrm{Tr}(\widetilde{M}_k \Lambda_k) & = \mathrm{Tr}(P M_k P \Lambda_k) \\
& = \mathrm{Tr}(M_k P \Lambda_k P) = \mathrm{Tr}(M_k \Lambda_k) = 1
\end{split}\end{equation*}
for $1 \leq k \leq d$. Also note that $| \lambda_1 \rangle, \cdots, | \lambda_d \rangle$ are of the form
\begin{equation*}
|\lambda_k \rangle = (I_d \otimes U_k) \ \frac{1}{\sqrt{d}} \sum_{j}^{d-1} |jj\rangle \ \ \ \ \ (1 \leq k \leq d)
\end{equation*}
where $I_d$ and $U_k$'s are respectively the identity and certain unitaries on span$\{|0\rangle, \cdots, |d-1\rangle\}$. Then
\begin{equation}\label{op1}
\Lambda_k^{\mathrm{T_A}} = (I_d \otimes U_k) \left( \frac{1}{d} \sum_{i, j = 0}^{d-1} |ij\rangle \langle ji| \right) (I_d \otimes U_k^{\dag})
\end{equation}
and hence there is
\begin{equation}\label{op2}
- \frac{1}{d} I_{d \otimes d} \ \leq \Lambda_k^{\mathrm{T_A}} \leq \ \frac{1}{d} I_{d \otimes d},
\end{equation}
where $I_{d \otimes d}$ is just the same thing as $P$. Here, the operators in (\ref{op1}) and (\ref{op2}) are all supported on subspace $\mathcal{H}_{d \otimes d}$ and $\sum_{i, j = 0}^{d-1} |ij\rangle \langle ji|$ has only eigenvalues $\pm 1$ on it. Therefore, by
\begin{equation}\label{ineq}
1 = \mathrm{Tr}(\widetilde{M}_k \Lambda_k) = \mathrm{Tr}(\widetilde{M}_k^{\mathrm{T_A}} \Lambda_k^{\mathrm{T_A}}) \leq \frac{\mathrm{Tr}(\widetilde{M}_k^{\mathrm{T_A}})}{d}
\end{equation}
we get $\mathrm{Tr}(\widetilde{M}_k) = \mathrm{Tr}(\widetilde{M}_k^{\mathrm{T_A}}) \geq d$ for $k \in \{1, \cdots, d\}$. The inequality in (\ref{ineq}) holds due to the fact that $\widetilde{M}_k^{\mathrm{T_A}} = P M_k^{\mathrm{T_A}} P \geq 0$. Now acting $P$ on both sides of (\ref{Mi}), we get
\begin{equation*}
\widetilde{M}_1 + \cdots + \widetilde{M}_{d+1} = I_{d \otimes d}
\end{equation*}
and therefore
\begin{equation*}
\mathrm{Tr}(\widetilde{M}_1) + \cdots + \mathrm{Tr}(\widetilde{M}_{d+1}) = d^2.
\end{equation*}
Since we have $\mathrm{Tr}(\widetilde{M}_k) \geq d$ for $k \in \{1, \cdots, d\}$, we get $\mathrm{Tr}(\widetilde{M}_{d+1}) = 0$ and thus $P M_{d+1} P = \widetilde{M}_{d+1} = 0$. Writing $M_{d+1}$ in its spectral decomposition form, one can further deduce $M_{d+1} P = P M_{d+1} = 0$ and therefore $M_{d+1} = (P + P^{\bot}) M_{d+1} (P + P^{\bot}) = P^{\bot} M_{d+1} P^{\bot}$. However, now we get:
\begin{equation*}\begin{split}
1 & = \mathrm{Tr}(M_{d+1} \Lambda_{d+1}) \\
& = \mathrm{Tr}(P^{\bot} M_{d+1} P^{\bot} \Lambda_{d+1}) \\
& = \mathrm{Tr}(M_{d+1} P^{\bot} \Lambda_{d+1} P^{\bot}) \\
& \leq  \mathrm{Tr}(P^{\bot} \Lambda_{d+1} P^{\bot}) \\
& \leq 1 - \langle e_{j_0} f_{j_0}| \Lambda_{d+1} |e_{j_0} f_{j_0}\rangle = 1 - \frac1d,
\end{split}\end{equation*}
which is an obvious contradiction. As a result, the states $|\lambda_{1}\rangle, \cdots, |\lambda_{d+1}\rangle$ are PPT-indistinguishable.
\end{proof}


\begin{thebibliography}{}
\bibitem{strongest nonlocality}
In some literature, the terminology ``strongest nonlocality'' is employed to characterize the absence of nontrivial OPLM in all bipartitions. However, such an expression may not be entirely suitable. \ In \cite{SBGhosh}, a new form of nonlocality has been introduced: local random authentication. \ According to its definition, a set of three Bell states is not nonlocal, while these states lack nontrivial OPLM.
\bibitem{subgroup}
When $t=1$, $H = \{[\mathbf{0}]\}$ is the trivial subgroup, which corresponds to the case of Proposition \ref{1d}.
\end{thebibliography}
\end{document}